\documentclass{amsproc}
\usepackage{amssymb,amsmath,amsthm}
\usepackage{epsfig}
\usepackage{hyperref}
\newcommand{\ie}{\emph{i.e.}}

\newcommand{\cf}{\emph{cf}}

\newcommand{\etal}{\emph{et al}}
\newcommand{\Real}{\mathbb{R}}
\newcommand{\Nat}{\mathbb{N}}

\newcommand{\sii}{L^2}
\newcommand{\sobi}{\mathop{W_0^{1,2}}\nolimits}
\newcommand{\Sobi}{\mathop{W^{1,2}}\nolimits}
\newcommand{\dSobi}{\mathop{W^{-1,2}}\nolimits}
\newcommand{\eps}{\varepsilon}
\newcommand{\Dom}{\mathfrak{D}}

\newcommand{\Hilbert}{\mathcal{H}}
\newcommand{\diag}{\mathrm{diag}}
\newcommand{\supp}{\mathop{\mathrm{supp}}\nolimits}
\numberwithin{equation}{section}
\newtheorem{Theorem}{Theorem}[section]
\newtheorem{Lemma}[Theorem]{Lemma}
\newtheorem{Corollary}[Theorem]{Corollary}
\theoremstyle{remark}
\newtheorem{Remark}[Theorem]{Remark}
\newtheorem*{Acknowledgment}{Acknowledgment}
\theoremstyle{definition}
\newtheorem{Definition}[Theorem]{Definition}
\begin{document}
%
%
\title{Twisting versus bending in quantum waveguides}
\author{David Krej\v{c}i\v{r}\'{\i}k}
\address{Department of Theoretical Physics, Nuclear Physics Institute,
Academy of Sciences, 250\,68 \v{R}e\v{z}, Czech Republic}
\email{krejcirik@ujf.cas.cz}
\thanks{
Research supported by FCT, Portugal,
through the grant SFRH/\-BPD/\-11457/\-2002,
and by the Czech Academy of Sciences and its Grant Agency
within the projects IRP AV0Z10480505 and A100480501,
and by the project LC06002 of the Ministry of Education,
Youth and Sports of the Czech Republic.
\bigskip \\
\emph{Date:} 25 March 2009.
\bigskip \\
This is a corrected version of the paper published in \\
\emph{Proc.\ Sympos.\ Pure Math., vol.~77, pp. 617--636,
Amer.\ Math.\ Soc., Providence, RI, 2008}.
\\
}
\begin{abstract}
\noindent
We make an overview of spectral-geometric effects
of twisting and bending in quantum waveguides
modelled by the Dirichlet Laplacian
in an unbounded three-dimensional tube
of uniform cross-section.
We focus on the existence of Hardy-type inequalities
in twisted tubes of non-circular cross-section.
\end{abstract}
\maketitle
\tableofcontents
%
%
%
\newpage
\section{Introduction}
%
The Dirichlet Laplacian in tubular domains
is a simple but remarkably successful model
for the quantum Hamiltonian in mesoscopic waveguide systems.
One of the main questions arising within the scope
of electronic transport
is whether or not there are geometrically induced bound states.
Indeed, some of the most important theoretical results in the field
are a number of theorems guaranteeing the existence
of stationary solutions to the Schr\"odinger equation
under rather simple and general physical conditions
(\cf~\cite{DE,LCM,KKriz} and references therein).

From the mathematical point of view,
one deals with a spectral-geometric problem
in quasi-cylindrical domains~\cite[Sec.~49]{Glazman},
or more generally in non-compact non-complete manifolds.
For such domains, in general,
the precise location of the essential spectrum
of the Dirichlet Laplacian is difficult,
and the existence of eigenvalues is a highly non-trivial fact.

The purpose of the present paper is to review
recent developments in the spectral theory
of a specific class of quantum waveguides
modelled by the Dirichlet Laplacian
in three-dimensional unbounded tubes of uniform cross-section.
We discuss how the spectrum depends upon
two independent geometric deformations: bending and twisting
(see Figure~\ref{f-tube}).
The attention is focused on improvements of the most recent results
as regards the existence of Hardy-type inequalities
in twisted tubes of non-circular cross-section;
new results and/or proofs are presented.

\begin{figure}[h]
\begin{center}
\includegraphics[width=\textwidth]{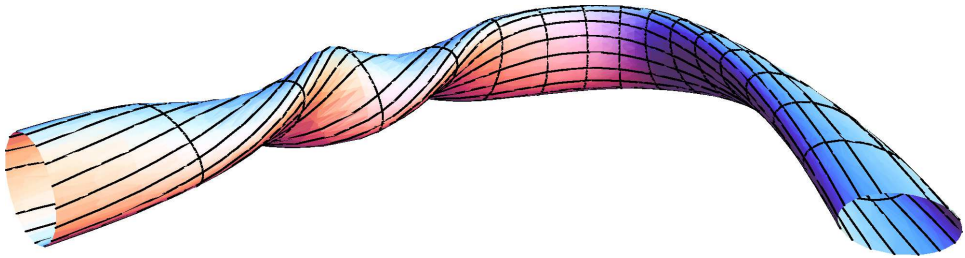}
\caption{An example of a tube of elliptical cross-section.
Twisting and bending are demonstrated on the left and right
part of the picture, respectively.}\label{f-tube}
\end{center}
\end{figure}

The outline of the paper is as follows.
In the following Section~\ref{Sec.geometry},
we introduce fundamental geometric objects
that will be used throughout the paper.
In particular, we identify a curved tube with a Riemannian manifold
and bring up the notions of bending and twisting.
In Section~\ref{Sec.Hamiltonian},
we identify the Dirichlet Laplacian in the tube
with the Laplace-Beltrami operator in the manifold.
Section~\ref{Sec.ess} is devoted to the location
of the essential spectrum under the hypotheses
that the twisting and bending vanish at infinity of the tube.

The effects of bending and twisting are studied
in Sections~\ref{Sec.bending} and~\ref{Sec.twisting}, respectively.
More precisely, in Section~\ref{Sec.bending}
we show that bending gives rise to eigenvalues
below the essential spectrum in non-twisted tubes,
while in Section~\ref{Sec.twisting}
we deal with the Hardy inequalities due to twisting in non-bent tubes.
As an application of the Hardy inequality,
in Section~\ref{Sec.mild} we show that the discrete spectrum
of a simultaneously bent and twisted tube is empty
provided that the bending is mild in a sense.
The paper is concluded in Section~\ref{Sec.end}
by referring to some open problems.

The source reference for Section~\ref{Sec.bending} (bending)
is my collaboration with Chen\-aud, Duclos and Freitas~\cite{ChDFK},
where the spectral theory for bent tubes of general cross-section
is performed for the first time.
However, the effect of bending has been known
for almost two decades \cite{ES,GJ,DE}.
On the other hand, except for some heuristic considerations~\cite{CB1,CB2},
the effect of twisting has been overlooked up to recently.
Section~\ref{Sec.twisting} (twisting)
and also Section~\ref{Sec.mild} are mainly based
on my collaboration with Ekholm and Kova\v{r}{\'\i}k~\cite{EKK},
but some of the ideas of~\cite{K3,FK3} are present too.
The reference list is likely far from being complete
and serves only this expos\'e.
\begin{Acknowledgment}
I am grateful to Hynek Kova\v{r}{\'\i}k and Jan K\v{r}{\'\i}\v{z}
for comments on an early draft of this paper.
\end{Acknowledgment}
%

\section{The geometry of a curved tube}\label{Sec.geometry}
%
\subsection{The reference curve}
Let $\Gamma:\Real\to\Real^3$ be a $C^3$-smooth curve
parameterized by its arc-length.
The curvature of~$\Gamma$ is defined by $\kappa:= |\ddot{\Gamma}|$.

Regarding~$\Gamma$ as a path of a unit-speed traveller in the space,
it is convenient to describe its motion
in a (non-inertial) reference frame moving along the curve.
One usually adopts the distinguished Frenet frame \cite[Sec.~1.2]{Kli},
\ie, the orthonormal triad of smooth vector fields $\{e_1,e_2,e_3\}$
-- called the tangent, normal and binormal vectors respectively --
defined by the prescriptions
\begin{equation}\label{Frenet}
  e_1 := \dot{\Gamma} \,, \qquad
  e_2 := \kappa^{-1} \ddot{\Gamma} \,, \qquad
  e_3 := e_1 \times e_2 \,.
\end{equation}
Here the cross denotes the vector product in~$\Real^3$.

Of course, one has to assume that the curvature
is never vanishing, \ie~$\kappa>0$,
in order to justify the construction~\eqref{Frenet}.
If~$\Gamma$ is a straight line (\ie~$\kappa=0$ identically),
one can choose a constant (inertial) Frenet frame instead.
Gluing constant Frenet frames to
the distinguished Frenet frame~\eqref{Frenet},
it is also possible to include curves satisfying $\kappa>0$
on a compact subset of~$\Gamma$ and being straight elsewhere
(\cf~\cite{EKK}), and others.
On the other hand, there exist infinitely smooth curves
with no smooth Frenet frame
(\cf~\cite[Chap.~1, p.~34]{Spivak2} for an example).

In any of the situations above when a global Frenet frame exists,
we say that the curve possesses an \emph{appropriate Frenet frame}.
Then the Frenet frame evolves along the curve via
the Serret-Frenet formulae~\cite[Sec.~1.3]{Kli}
\begin{equation}\label{Serret}
  \begin{pmatrix}
    e_1 \\ e_2 \\ e_3
  \end{pmatrix}^{\textbf{.}}
  =
  \begin{pmatrix}
    0 & \kappa & 0 \\
    -\kappa & 0 & \tau \\
    0 & -\tau & 0
  \end{pmatrix}
  \!
  \begin{pmatrix}
    e_1 \\ e_2 \\ e_3
  \end{pmatrix}
  ,
\end{equation}
where~$\tau$ is the torsion of~$\Gamma$,
actually defined by~\eqref{Serret}.

\subsection{The general moving frame}
Let $\theta:\Real\to\Real$ be a $C^1$-smooth function.
We define a rotation matrix-valued function
$\mathcal{R}^\theta:\Real \to \mathrm{SO}(3)$ by setting
\begin{equation*}
  \mathcal{R}^\theta
  :=
  \begin{pmatrix}
    1 & 0 & 0 \\
    0 & \cos\theta & -\sin\theta \\
    0 & \sin\theta & \cos\theta
  \end{pmatrix}
  .
\end{equation*}
Finally, we introduce a new moving frame
$\{e_1^\theta,e_2^\theta,e_3^\theta\}$ along~$\Gamma$
by rotating normal components of the appropriate Frenet frame
$\{e_1,e_2,e_3\}$ of~$\Gamma$ by the angle function~$\theta$:
\begin{equation}\label{frame}
  e^\theta_i := \sum_{j=1}^3 \mathcal{R}^\theta_{ij} \, e_j
  \,, \qquad
  i \in \{1,2,3\}
  \,.
\end{equation}

Using~\eqref{Serret}, it is easy to check
that the new frame evolves along the curve via
\begin{equation}\label{motion}
  \begin{pmatrix}
    e_1^\theta \\ e_2^\theta \\ e_3^\theta
  \end{pmatrix}^{\textbf{.}}
  =
  \begin{pmatrix}
    0 & \kappa\cos\theta & \kappa\sin\theta \\
    -\kappa\cos\theta & 0 & \tau-\dot{\theta} \\
    -\kappa\sin\theta & -(\tau-\dot{\theta}) & 0
  \end{pmatrix}
  \!
  \begin{pmatrix}
    e_1^\theta \\ e_2^\theta \\ e_3^\theta
  \end{pmatrix}
  .
\end{equation}
\subsection{The cross-section}
Let~$\omega$ be a bounded open connected set in $\Real^2$.
We do not assume any regularity conditions about the boundary~$\partial\omega$.
It is convenient to introduce the quantity
\begin{equation*}
  a:=\sup_{t\in\omega}|t| \,,
\end{equation*}
measuring the distance of the farthest point of~$\overline{\omega}$ to the origin.
We say that~$\omega$ is
\emph{rotationally invariant with respect to the origin}
if
\begin{equation*}
  \forall \vartheta \in (0,2\pi) \,, \qquad
  \omega_\vartheta :=
  \bigg\{ \Big(
  \sum_{j=2}^3 t_j\,\mathcal{R}^\vartheta_{j2} \,,
  \sum_{j=2}^3 t_j\,\mathcal{R}^\vartheta_{j3}
  \Big) \
  \bigg|\ (t_2,t_3)\in\omega \bigg\}
  = \omega
  \,,
\end{equation*}
with the natural convention that we identify~$\omega$ and~$\omega_\vartheta$
(and other open sets)
provided that they differ on a set of zero capacity.
Hence, modulus a set of zero capacity,
$\omega$~is rotationally symmetric
with respect to the origin in~$\Real^2$ if, and only if,
it is a disc or an annulus centered at the origin of~$\Real^2$.

\subsection{The tube}
A tube~$\Omega$ is defined by moving the cross-section~$\omega$
along the reference curve~$\Gamma$
together with a generally rotated frame~\eqref{frame}.
More precisely, we set
\begin{equation*}
  \Omega := \mathcal{L}(\Real\times\omega)
  \,,
\end{equation*}
where $\mathcal{L}$ is the mapping from
the straight tube $\Real\times\omega$ to~$\Real^3$
defined by
\begin{equation}\label{tubemap}
  \mathcal{L}(s,t):=\Gamma(s) + \sum_{j=2}^3 t_j \, e_j^\theta(s)
  \,.
\end{equation}
Since the curve~$\Gamma$ can be reconstructed from
the curvature functions~$\kappa$ and~$\tau$
(\cf~\cite[Thm.~1.3.6]{Kli}),
the tube~$\Omega$ is fully determined by giving
the cross-section~$\omega$ (including its position in~$\Real^2$)
and the triple of functions $\kappa$, $\tau$ and~$\theta$.

Our strategy to deal with the curved geometry of the tube
is to identify~$\Omega$ with the Riemannian manifold
$(\Real\times\omega,G)$, where $G=(G_{ij})$ is the metric tensor
induced by the embedding~$\mathcal{L}$, \ie,
$$
  G_{ij}:=(\partial_i\mathcal{L})\cdot(\partial_j\mathcal{L})
  \,, \qquad
  i,j \in \{1,2,3\}
  \,.
$$
Here the dot denotes the scalar product in~$\Real^3$.
In other words, we parameterize~$\Omega$ globally
by means of the ``coordinates'' $(s,t)$ of~\eqref{tubemap}.
To this aim, we need to impose natural restrictions
in order to ensure that~$\mathcal{L}$ induces a $C^1$-diffeomorphism
between $\Real\times\omega$ and~$\Omega$.

Using~\eqref{motion}, we find
\begin{equation}\label{metric}
  G =
  \begin{pmatrix}
    h^2+h_2^2+h_3^2  & h_2 & h_3 \\
    h_2 & 1 & 0 \\
    h_3 & 0 & 1 \\
  \end{pmatrix}
  , \quad
  \begin{aligned}
  h(s,t)
  &:= 1 - [t_2\cos\theta(s)+t_3\sin\theta(s)] \, \kappa(s) \,,
  \\
  h_2(s,t)
  &:= - t_3 \, [\tau(s)-\dot{\theta}(s)] \,,
  \\
  h_3(s,t)
  &:= t_2 \, [\tau(s)-\dot{\theta}(s)] \,.
\end{aligned}
\end{equation}
Consequently,
$$
  |G| := \det(G) = h^2 \,.
$$
By virtue of the inverse function theorem,
the mapping~$\mathcal{L}$ induces a local $C^1$-diffeomorphism
provided that the Jacobian~$h$ does not vanish on $\Real\times\omega$.
In view of the uniform bounds
\begin{equation}\label{1<G<1}
  0 <
  1 - a \, \|\kappa\|_{L^\infty(\Real)}
  \ \leq \ h \ \leq \
  1 + a \, \|\kappa\|_{L^\infty(\Real)}
  < \infty
  \,,
\end{equation}
the positivity of~$h$ is guaranteed by the hypothesis
\begin{equation}\label{Ass.basic1}
  \kappa \in L^\infty(\Real)
  \qquad\mbox{and}\qquad
  a \, \|\kappa\|_{L^\infty(\Real)} < 1
  \,.
\end{equation}
The mapping then becomes a global diffeomorphism if,
in addition to~\eqref{Ass.basic1}, we assume that
\begin{equation}\label{Ass.basic2}
  \mathcal{L} \quad \mbox{is injective}
  \,.
\end{equation}
For sufficient conditions ensuring~\eqref{Ass.basic2}
we refer to~\cite[App.]{EKK}.

\subsection{The natural hypotheses}
For the convenience of the reader,
we summarize here characteristic conditions
needed for the construction of a tube~$\Omega$:

\smallskip
\noindent
\mbox{\emph{1.~the reference curve~$\Gamma$ is $C^3$-smooth
and possesses an appropriate Frenet frame;}}%

\smallskip
\noindent
\emph{2.~the cross-section~$\omega$ is bounded;}

\smallskip
\noindent
\emph{3.~the angle function~$\theta$ is $C^1$-smooth;}

\smallskip
\noindent
\emph{4.~\eqref{Ass.basic1} and \eqref{Ass.basic2} hold
(\ie, $\Omega$ is not self-intersecting).}

\smallskip
\noindent
These hypotheses will be assumed henceforth,
without any further repetitions.

\begin{Remark}
Relaxing the geometrical interpretation of~$\Omega$
being a non-self-inter\-secting tube in~$\Real^3$,
it is possible to consider $(\Real\times\omega,G)$
as an abstract Riemannian manifold
where only the reference curve~$\Gamma$ is embedded in~$\Real^3$.
Then one does not need to assume~\eqref{Ass.basic2},
and the spectral results below hold in this more general situation, too.
\end{Remark}
\subsection{The definitions of bending and twisting}\label{Sec.defs}
It is clear from the equations of motion
of the general moving frame~\eqref{motion}
that there are two independent geometric effects in curved tubes.
\begin{Definition}[bending]\label{Def.bending}
The tube~$\Omega$ is said to be \emph{bent} if, and only if,
the reference curve~$\Gamma$ is not a straight line,
\ie, $\kappa \not= 0$.
\end{Definition}
\begin{Definition}[twisting]\label{Def.twisting}
The tube~$\Omega$ is said to be \emph{twisted} if, and only if,
the cross-section~$\omega$ is not rotationally invariant
with respect to the origin
and \mbox{$\tau-\dot{\theta} \not= 0$}.
\end{Definition}

Since our class of tubes is such that the cross-section~$\omega$
is locally perpendicular to the tangent vector of~$\Gamma$,
it is easy to check that
the Definitions~\ref{Def.bending} and~\ref{Def.twisting}
are independent of the ``parametrization'' of~$\Omega$,
\ie~the possibly different choice of the reference curve~$\Gamma$,
the position of~$\omega$ in~$\Real^2$ and the function~$\theta$
leading to the same shape of~$\Omega$.

Of course, in the second definition it is necessary to assume that
the cross-section is not rotationally invariant with respect to the origin,
since the shape of the tube~$\Omega$ is not influenced
by a special choice of~$\theta$ if the cross-section
is a disc or an annulus centered at the reference line.

The message of the first definition is clear:
the reference curve must be non-trivially curved
to give rise to a bending of the tube.
On the other hand, the requirement $\tau-\dot{\theta} \not= 0$
is less intuitive in the definition of twisting,
unless~$\Gamma$ is straight.
Therefore we point out the equivalence
of the following statements:

\smallskip
\noindent
1.~\underline{$\tau-\dot{\theta} = 0$}.
Regarding this identity as a differential equation for~$\theta$,
its solution leads to a special choice
of the moving frame~\eqref{frame} along~$\Gamma$,
unique up to initial conditions.
This special frame is known as the Tang frame
in the physical literature \cite{Tsao-Gambling_1989}.

\smallskip
\noindent
2.~\underline{No transverse rotations}.
Let us regard the parallel curve
$
  s \mapsto \Gamma^\theta(s) := \mathcal{L}\big(s,a,0\big)
$
as the path of a traveller in the space.
Calculating its velocity
$$
  \dot{\Gamma}^\theta
  = (1-a\,\kappa\cos\theta) \, e_1
  + a \, (\tau-\dot{\theta}) \, e_3^\theta
  \,,
$$
we see that its motion is non-inertial
unless~$\Gamma$ is a straight line and~$\theta$ is a constant.
From the formula we also conclude that
the component of the angular velocity
of $\Gamma^\theta$ relative to~$\Gamma$
that corresponds to the instantaneous rotations of~$\Gamma^\theta$
about the tangent~$e_1$ is precisely $\tau-\dot{\theta}$.

\smallskip
\noindent
3.~\underline{Orthogonality}.
From the expression for the metric~\eqref{metric},
it is readily seen that the ``coordinates'' $(s,t)$
of~\eqref{tubemap} are orthogonal if, and only if,
$\tau-\dot{\theta} = 0$ holds.
This makes the Tang frame a technically useful choice
for the tubes with cross-sections rotationally invariant
with respect to the origin.

\smallskip
\noindent
4.~\underline{Zero intrinsic curvature}.
Let~$\Sigma$ be the ruled surface generated
by the vectors~$e_2^\theta$ along~$\Gamma$ (see Figure~\ref{f-strip}),
\ie, $\Sigma:=\mathcal{L}\big(\Real\times(0,a)\times\{0\}\big)$.
Then the item~1 is equivalent to the fact that
the Gauss curvature of~$\Sigma$ vanishes identically
(\cf~\cite[Sec.~2]{K3}).

\begin{figure}[h]
\begin{center}
\includegraphics[width=\textwidth]{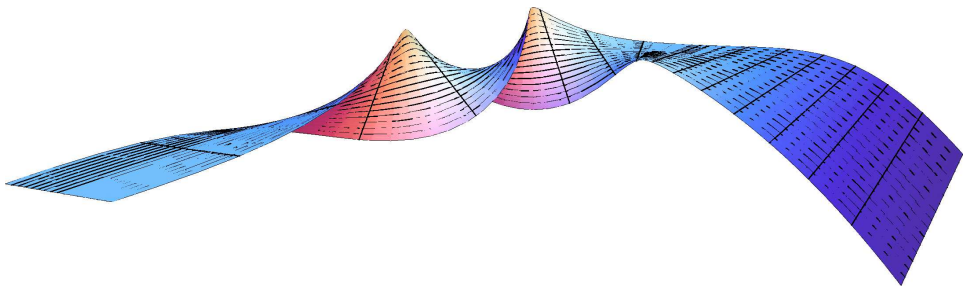}
\caption{The ruled surface~$\Sigma$ associated
with the tube of Figure~\ref{f-tube}.}\label{f-strip}
\end{center}
\end{figure}

\smallskip
\noindent
5.~\underline{Parallel transport}.
It is a well known fact~\cite[Prop.~3.7.5]{Kli}
that the Gauss curvature of a ruled surface vanishes identically
if, and only if, the surface is developable, \ie,
the surface normal vector field is a constant along generators.

\smallskip
\noindent
6.~\underline{Zero Berry phase}.
\footnote{
I am grateful to Yves Colin de Verdi{\`e}re
for pointing out this equivalence
during my talk in Cambridge~\cite{K-Cambridge}.
}
Equipping the normal fiber of~$\Gamma$ with its Berry (geometric) connection,
the Tang frame is the only frame whose transport along~$\Gamma$ is parallel
(\cf~\cite[Sec.~5]{CdeV_2006}).
Alternatively, by~\cite[Prop.~1]{CdeV_2006},
the Berry connection of~$\Sigma$ is equal to its Levi-Civita connection
whose curvature is the Gauss curvature of~$\Sigma$
(see the item~4 above).

\section{The Hamiltonian}\label{Sec.Hamiltonian}
%
\subsection{The initial Laplacian}
Let us recall that,
under the hypotheses~\eqref{Ass.basic1} and~\eqref{Ass.basic2},
the tube~$\Omega$ is an open subset of~$\Real^3$.
Hence, the corresponding Dirichlet Laplacian
can be introduced in a standard way
as the self-adjoint operator~$-\Delta_D^\Omega$ in $\sii(\Omega)$
associated with the quadratic form
\begin{equation*}
  Q_D^\Omega[\Psi] := \|\nabla \Psi\|_{\sii(\Omega)}^2
  \,, \qquad
  \Psi \in \Dom(Q_D^\Omega) := \sobi(\Omega)
  \,.
\end{equation*}
\subsection{The Laplacian in curvilinear coordinates}
Our strategy to investigate~$-\Delta_D^\Omega$
is to express it in the coordinates determined by~\eqref{tubemap}.
More specifically, recalling the diffeomorphism
between $\Real\times\omega$ and~$\Omega$ given by~$\mathcal{L}$,
we can identify the Hilbert space $\sii(\Omega)$ with
\begin{equation}\label{Hilbert}
  \Hilbert := \sii(\Real\times\omega,G)
  \,.
\end{equation}
The latter means the usual $\sii$-space over $\Real\times\omega$
equipped with the inner product
$$
  (\psi,\phi)_\Hilbert
  := \int_{\Real\times\omega} \overline{\psi(s,t)} \, \phi(s,t) \
  h(s,t) \, ds \, dt
  \,.
$$

Using the usual differential-geometric calculus,
we then get that the Laplacian~$-\Delta_D^\Omega$ is unitarily equivalent
to the operator~$H$ in~$\Hilbert$ associated
with the quadratic form
\begin{equation}\label{form}
  Q[\psi] := \left(
  \partial_i\psi, G^{ij} \, \partial_j\psi
  \right)_\Hilbert
  \,, \qquad
  \psi \in \Dom(Q) := \sobi(\Real\times\omega,G)
  \,.
\end{equation}
Here and hereafter we adopt the Einstein summation convention,
the range of indices being $1,2,3$.
$G^{ij}$ stands for the coefficients of the matrix inverse to~$G$:
\begin{equation*}
  G^{-1} = \frac{1}{h^2}
  \begin{pmatrix}
    1 & -h_2 & -h_3 \\
    -h_2 & h^2+h_2^2 & h_2 h_3 \\
    -h_3 & h_3 h_2 & h^2+h_3^2 \\
  \end{pmatrix}
  .
\end{equation*}
Finally, $\sobi(\Real\times\omega,G)$ denotes the completion
of $C_0^\infty(\Real\times\omega)$ with respect to the norm
$
  \|\cdot\|_{\Dom(Q)} :=
  \left(
  Q[\cdot] + \|\cdot\|_\Hilbert^2
  \right)^{1/2}
$.
If the functions~$\kappa$ and $\tau-\dot\theta$ are bounded,
then the $\Dom(Q)$-norm is equivalent to the usual norm
in $\Sobi(\Real\times\omega)$.

\section{Stability of the essential spectrum}\label{Sec.ess}
%
If the tube is straight, \ie, $\kappa=0=\tau-\dot\theta$,
the metric~\eqref{metric} reduces to the Euclidean metric
and the spectrum of the Laplacian can be found easily
by ``separation of variables'':
\begin{equation*}
  \sigma(-\Delta_D^{\Real\times\omega})
  = \sigma_\mathrm{ess}(-\Delta_D^{\Real\times\omega})
  = [E_1,\infty)
  \,, \qquad\mbox{where}\qquad
  E_1 := \inf\sigma(-\Delta_D^\omega)
\end{equation*}
denotes the first Dirichlet eigenvalue in the cross-section~$\omega$.
The positive normalized eigenfunction of~$-\Delta_D^\omega$
corresponding to~$E_1$ will be denoted by~$\mathcal{J}_1$.

The essential spectrum of the Laplacian in a manifold is
determined by the behaviour of the metric at infinity
(and possibly at the boundary) only.
Inspecting the dependence of the coefficients of~\eqref{metric}
on large ``longitudinal distances''~$s$,
in our case it is natural to expect that the interval $[E_1,\infty)$
will form the essential spectrum of~$-\Delta_D^\Omega$ as well,
provided that
\begin{equation}\label{Ass.decay}
  \lim_{|s|\to\infty} \kappa(s) = 0
  \qquad\mbox{and}\qquad
  \lim_{|s|\to\infty} [\tau(s)-\dot\theta(s)] = 0
  \,.
\end{equation}
\begin{Theorem}\label{Thm.ess}
Under the hypotheses~\eqref{Ass.decay},
$$
  \sigma_\mathrm{ess}(-\Delta_D^\Omega) = [E_1,\infty)
  \,.
$$
\end{Theorem}
\begin{proof}
The proof consists of two steps.

\smallskip
\noindent
1.~\underline{$\inf\sigma_\mathrm{ess}(H) \geq E_1$}.
We impose additional Neumann boundary conditions
on the tube cross-sections at $|s|=s_0>0$.
This leads to a direct sum of Laplacians
$H_\mathrm{int}^N$ and $H_\mathrm{ext}^N$ in the subsets
$I_\mathrm{int}\times\omega$
and $I_\mathrm{ext}\times\omega$, respectively,
where $I_\mathrm{int}:=(-s_0,s_0)$
and $I_\mathrm{ext}:=(-\infty,-s_0)\cup(s_0,\infty)$.
More precisely,
$H_\mathrm{int}^N$ is introduced as the operator in
$\Hilbert_\mathrm{int}:=\sii(I_\mathrm{int}\times\omega,G)$
associated with the quadratic form
$$
  Q_\mathrm{int}^N[\psi] := \left(
  \partial_i\psi, G^{ij} \, \partial_j\psi
  \right)_{\Hilbert_\mathrm{int}}
  , \quad
  \psi \in \Dom(Q_\mathrm{int}^N)
  :=
  \big\{
  \psi \!\upharpoonright\! (I_\mathrm{int}\times\omega)
  \ | \
  \psi \in \Dom(Q)
  \big\}
  \,,
$$
and similarly for $H_\mathrm{ext}^N$.
$H_\mathrm{int}^N$ is an operator with compact resolvent.
By the minimax principle, the essential spectrum of~$H$ is
estimated from below by the lowest point
in the essential spectrum of~$H_\mathrm{ext}^N$.
Using the crude bound
$
  G^{-1} \geq \diag(0,1,1)
$,
we get
\begin{align*}
  Q_\mathrm{ext}^N[\psi] &\geq \|\nabla'\psi\|_{\Hilbert_\mathrm{ext}}^2
  \geq (\inf h) \, \|\nabla'\psi\|_{\sii(I_\mathrm{ext}\times\omega)}^2
  \geq E_1 \, (\inf h) \, \|\psi\|_{\sii(I_\mathrm{ext}\times\omega)}^2
  \\
  &\geq E_1 \, \frac{\inf h}{\sup h} \, \|\psi\|_{\Hilbert_\mathrm{ext}}^2
\end{align*}
for all $\psi \in \Dom(Q_\mathrm{ext}^N)$,
where the infima and suprema are taken over $I_\mathrm{ext}\times\omega$,
and $\nabla':=(\partial_2,\partial_3)$.
Using the first of the hypotheses~\eqref{Ass.decay},
we see that the spectrum of~$H_\mathrm{ext}^N$
is estimated from below by~$E_1$ times a function of~$s_0$
tending to~$1$ as $s_0 \to \infty$.
We conclude with noticing that the essential spectrum is a closed set
and that~$s_0$ can be chosen arbitrarily large.

\smallskip
\noindent
2.~\underline{$\sigma(H) \supseteq [E_1,\infty)$}.
It suffices to construct for each $k \in [0,\infty)$
a sequence $\{\psi_n\}_{n=1}^\infty$ from $\Dom(H)$,
with elements normalized to~$1$ in~$\Hilbert$, such that
$[H-(k^2+E_1)]\psi_n \to 0$ in $\Hilbert$ as $n \to \infty$.
One is tempted to use a tensor product
of plane waves ``localized at infinity''
and the first Dirichlet eigenfunction~$\mathcal{J}_1$
in the cross-section.
Namely, let us set
$$
  \phi_n(s,t) := \varphi_n(s) \, e^{i k s} \, \mathcal{J}_1(t)
  \,,
$$
where $\varphi_n(s):=\varphi(n^{-1}s-n)$ with~$\varphi$ being
a non-zero $C^\infty$-smooth function with a compact support in $(-1,1)$.
Under additional assumptions
about the decay of~$\kappa$ and $\tau-\dot\theta$ at infinity
(involving derivatives),
it is indeed possible to show that
$\psi_n := \phi_n/\|\phi_n\|_\Hilbert$
form the desired sequence.

To avoid the additional assumptions, however,
we reconsider~$H$ as an operator with domain $\Dom(Q)$
in the Hilbert space
$[\Dom(Q)]^*=\dSobi(\Real\times\omega,G)$,
the topological dual of $\Dom(Q)$ equipped with the dual norm.
This is justified by the fact that
$H+1:\Dom(H)\to\Hilbert$ and $H+1:\Dom(Q)\to[\Dom(Q)]^*$ are isomorphisms.
The functions $\phi_n$ clearly belong to $\Dom(Q)$.
Let us show that~$\phi_n$ can be renormalized to~$1$ in $[\Dom(Q)]^*$, \ie,
\begin{equation}\label{check1}
  \lim_{n\to\infty} \|\phi_n\|_{[\Dom(Q)]^*} > 0
\end{equation}
and that the renormalized functions
$\xi_n := \phi_n/\|\phi_n\|_{[\Dom(Q)]^*}$ for large~$n$
form the desired sequence in $[\Dom(Q)]^*$, \ie,
\begin{equation}\label{check2}
  \lim_{n\to\infty} \big\| [H-(k^2+E_1)]\xi_n \big\|_{[\Dom(Q)]^*} = 0
  \,.
\end{equation}

It easy to verify that
$$
  | h - 1 | \leq a |\kappa| \,, \qquad
  | G^{-1} h -1 | \leq C \big( |\kappa|+|\tau-\dot\theta| \big) \, 1
  \,.
$$
Here the second inequality holds in the sense of matrices,
$1$~denotes the identity matrix,
and~$C$ is a positive constant depending on~$a$
and the supremum norms of~$\kappa$ and $\tau-\dot\theta$.
Consequently,
\begin{align*}
  \big|
  (\phi,\phi_n)_\Hilbert
  - (\phi,\phi_n)_{\sii(\Real\times\omega)}
  \big|
  &\leq C \, b_n \, \|\phi\|_{\sii(\Real\times\omega)} \,
  \|\phi_n\|_{\sii(\Real\times\omega)}
  \,,
  \\
  \big|
  Q(\phi,\phi_n) - (\nabla\phi,\nabla\phi_n)_{\sii(\Real\times\omega)}
  \big|
  &\leq C \, b_n \, \|\nabla\phi\|_{\sii(\Real\times\omega)} \,
  \|\nabla\phi_n\|_{\sii(\Real\times\omega)}
  \,,
\end{align*}
for any $\phi \in \Dom(Q)$,
where the constant~$C$ possibly differs from that above and
$$
  b_n := \sup_{\supp\varphi_n}\big( |\kappa|+|\tau-\dot\theta| \big)
  \xrightarrow[n\to\infty]{} 0
  \,.
$$
Integrating by parts, using the explicit expression
$$
  - \Delta \phi_n(s,t)
  = (k^2+E_1)\,\phi_n(s,t)
  -2ik \, \dot{\varphi}_n(s) \, e^{iks} \, \mathcal{J}_1(t)
  -\ddot{\varphi}_n(s) \, e^{iks} \, \mathcal{J}_1(t)
$$
and recalling that~$\mathcal{J}_1$ is normalized to~$1$ in $\sii(\omega)$,
we can estimate
\begin{multline*}
  \big|
  (\nabla\phi,\nabla\phi_n)_{\sii(\Real\times\omega)}
  - (k^2+E_1) \big(\phi,\phi_n\big)_{\sii(\Real\times\omega)}
  \big|
  \\
  \leq
  \|\phi\|_{\sii(\Real\times\omega)}
  \left(
  2 k \, \|\dot{\varphi}_n\|_{\sii(\Real)}
  + \|\ddot{\varphi}_n\|_{\sii(\Real)}
  \right)
\end{multline*}
for any $\phi \in \Dom(Q)$.
Finally,
$
  \|\phi_n\|_{\sii(\Real\times\omega)}^2
  = \|\varphi_n\|_{\sii(\Real)}^2
  = n \, \|\varphi\|_{\sii(\Real)}^2
$,
and
$$
  \frac{\|\dot{\varphi}_n\|_{\sii(\Real)}}{\|\varphi_n\|_{\sii(\Real)}}
  = \frac{1}{n} \frac{\|\dot{\varphi}\|_{\sii(\Real)}}{\|\varphi\|_{\sii(\Real)}}
  \xrightarrow[n\to\infty]{} 0
  \,, \qquad
  \frac{\|\ddot{\varphi}_n\|_{\sii(\Real)}}{\|\varphi_n\|_{\sii(\Real)}}
  = \frac{1}{n^2} \frac{\|\ddot{\varphi}\|_{\sii(\Real)}}{\|\varphi\|_{\sii(\Real)}}
  \xrightarrow[n\to\infty]{} 0
  \,.
$$

Using these preliminaries, \eqref{check1}~follows from
$$
  \|\phi_n\|_{[\Dom(Q)]^*}
  \equiv \sup_{\phi\in\Dom(Q)\setminus\{0\}}
  \frac{|(\phi,\phi_n)_\Hilbert|}{\|\phi\|_{\Dom(Q)}}
  \geq \frac{\|\phi_n\|_\Hilbert^2}{\,\|\phi_n\|_{\Dom(Q)}}
  = \frac{\|\phi_n\|_\Hilbert}{\sqrt{
  1+\frac{Q[\phi_n]}{\,\|\phi_n\|_\Hilbert^2}
  }}
  \xrightarrow[n\to\infty]{} \infty
  \,,
$$
while~\eqref{check2} is a consequence of the formula
\begin{align*}
  \big\| [H-(k^2+E_1)]\xi_n \big\|_{[\Dom(Q)]^*}
  &\equiv \sup_{\phi\in\Dom(Q)\setminus\{0\}}
  \frac{|Q(\phi,\xi_n)-(k^2+E_1)(\phi,\xi_n)_\Hilbert|}{\|\phi\|_{\Dom(Q)}}
\end{align*}
and some elementary estimates.
\end{proof}

While the Neumann bracketing in the first step
of the proof of Theorem~\ref{Thm.ess} is standard,
the idea in the second step is remarkable.
It enables one to study spectral properties of~$H$
by working with its quadratic form only.
In the context of Theorem~\ref{Thm.ess},
it is an alternative to the Weyl-type characterization
of essential spectrum adapted to the quadratic-form setting
by Iftimie \etal\/ \cite[Lem.~4.1]{DDI}.

For three-dimensional tubes of cross-section
being a disc centered at the reference curve,
satisfying additional assumptions about the decay of curvature at infinity,
Theorem~\ref{Thm.ess} has been proved previously
by Goldstone and Jaffe~\cite{GJ} (compactly supported~$\kappa$)
and by Duclos and Exner~\cite{DE} (additional vanishing
of~$\dot\kappa$ and $\ddot{\kappa}$ at infinity).
Of course, if~$\omega$ is rotationally invariant with respect to the origin,
we can choose the Tang frame
as the moving frame along~$\Gamma$ without loss of generality.
Consequently, the second hypothesis of~\eqref{Ass.decay}
is superfluous in these situations.

In the general case, however,
both the assumptions~\eqref{Ass.decay} are important.
For instance, if the tube is periodically bent but not twisted
(respectively periodically twisted but not bent),
it follows from Theorem~\ref{Thm.bending}
(respectively Corollary~\ref{Cor.uniform}) below
that the essential spectrum starts strictly below
(respectively above) the energy~$E_1$.

For non-twisted tubes of general cross-section (and of arbitrary dimension),
assuming just the vanishing of~$\kappa$ at infinity,
Theorem~\ref{Thm.ess} was proved for the first time in~\cite{ChDFK}
using the ideas of~\cite[Sec.~4]{DDI}.
The proof is different from the present one
but easily adaptable to the twisted case too.

\section{The effect of bending}\label{Sec.bending}
%
It turns out that bending acts as an attractive interaction
in the sense that it gives rise to a spectrum below the energy~$E_1$.
\begin{Theorem}\label{Thm.bending}
Let $\kappa \not= 0$ and $\tau-\dot\theta = 0$.
Then
$$
  \inf\sigma(-\Delta_D^\Omega) < E_1
  \,.
$$
\end{Theorem}
\begin{proof}
The proof is variational, based on the fact that
the Rayleigh quotient of the operator $H-E_1$
can be made negative for a trial function
built from $(s,t)\mapsto\mathcal{J}_1(t)$,
a generalized eigenfunction of~$H$ corresponding to~$E_1$.

\smallskip
\noindent
1.~First one verifies that
$$
  Q_1[\psi_n] := Q[\psi_n] - E_1 \|\psi_n\|_\Hilbert^2
  \longrightarrow 0
$$
as $n\to\infty$
for the choice $\psi_n(s,t):=\varphi_n(s)\mathcal{J}_1(t)$,
where $\varphi_n:\Real \to [0,1]$ form a sequence
of functions from $\Sobi(\Real)$ such that
$\varphi_n(s) \to 1$ for a.e.\ $s\in\Real$
and $\|\dot\varphi_n\|_{\sii(\Real)} \to 0$ as $n \to \infty$.

\smallskip
\noindent
2.~Then one shows that adding a small perturbation $\phi$ to~$\psi_n$,
the form $Q_1[\psi_n+\phi]$ can be made negative for large~$n$.
One can take, for instance,
$$
  \phi(s,t) :=
  \eps \, \xi(s) \, \left[t_2\cos\theta(s)+t_3\sin\theta(s)\right] \,
  \mathcal{J}_1(t)
  \,,
$$
where~$\eps$ is a real number of suitable sign
and $\xi \in \Sobi(\Real)\setminus\{0\}$ is a non-negative function
with a compact support contained in an interval where~$\kappa$
is not zero and does not change sign.

\smallskip
\noindent
We refer to~\cite[Sec.~3.2]{ChDFK} for more details.
\end{proof}

The original idea of the proof to build a test function from
the generalized eigenfunction corresponding
to the threshold of the essential spectrum~$E_1$
belongs to Goldstone and Jaffe~\cite{GJ}.
In their paper Theorem~\ref{Thm.bending} was proved,
under the additional assumption that~$\kappa$ is compactly supported,
for tubes of cross-section being a disc centered at the reference curve.
Duclos and Exner~\cite{DE} made the proof of~\cite{GJ} rigorous
and relaxed the condition about the compact support of~$\kappa$,
however, technical assumptions about the local behaviour
of~$\dot\kappa$ and $\ddot{\kappa}$ had to be imposed.
The generalization to tubes of general cross-sections
(and of arbitrary dimension) was made in~\cite{ChDFK},
$\kappa \not= 0$ and $\tau-\dot\theta = 0$
being the only important assumptions.

As a consequence, we get
\begin{Corollary}
Under the hypotheses of Theorems~\ref{Thm.ess} and~\ref{Thm.bending},
$$
  \sigma_\mathrm{disc}(-\Delta_D^\Omega) \not= \varnothing
  \,.
$$
\end{Corollary}
%

\section{The effect of twisting}\label{Sec.twisting}
%
Now we come to the most recent results in the theory of quantum waveguides.
It turns out that the effect of twisting is quite opposite
to that of bending: it acts rather as a repulsive interaction.
This statement should be understood in a vague sense that
twisting tends to rise the spectral threshold of $-\Delta_D^\Omega$.
However, a more subtle approach is necessary
to analyse the effect of twisting rigorously,
in particular if the tube is asymptotically straight~\eqref{Ass.decay}.
By analogy with Schr\"odinger operators,
the technique of Hardy inequalities seems to be adequate here.

Throughout this section we assume that the tube~$\Omega$
is not bent, \ie~$\kappa=0$
(then~$\Gamma$ is a straight line and we also have $\tau=0$).
Under this condition,
the Hilbert space~\eqref{Hilbert} reduces to $\sii(\Real\times\omega)$.
The quadratic form~\eqref{form} can be identified with
\begin{align*}
  Q_{\alpha}^I[\psi] &:=
  \|\partial_1\psi-\alpha\,\partial_u\psi\|_{\sii(I\times\omega)}^2
  + \|\nabla'\psi\|_{\sii(I\times\omega)}^2
  \,,
  \\
  \psi \in \Dom(Q_{\alpha}^I) &:=
  \big\{
  \psi\!\upharpoonright\! (I\times\omega)
  \ | \ \psi\in\sobi(\Real\times\omega)
  \big\}
  \,,
\end{align*}
with the choice $\alpha=\dot{\theta}$ and $I=\Real$.
Here $\nabla':=(\partial_2,\partial_3)$
and $\partial_u$ denotes the transverse angular-derivative operator
$$
  \partial_u := u \cdot \nabla'
  \qquad\mbox{with}\qquad
  u(t):=(t_3,-t_2)
$$
and the dot being the scalar product in~$\Real^2$.
We proceed in a greater generality by assuming
that $\alpha:I\to\Real$ is an arbitrary bounded function
(we denote by the same letter the function $\alpha \otimes 1$
on $I\times\omega$)
and that $I\subseteq\Real$ is an arbitrary open interval.
Let~$H_\alpha^I$ be the self-adjoint operator in $\sii(I\times\omega)$
associated with~$Q_\alpha^I$.

\subsection{A Poincar\'e-type inequality}
%
Let $\lambda(\alpha,I)$ denote the spectral threshold of
the shifted operator $H_\alpha^I-E_1$,
\ie\ the lowest point in its spectrum.
If the interval~$I$ is bounded, then the spectrum of $H_\alpha^I$
is purely discrete and $\lambda(\alpha,I)$
is just the first eigenvalue of $H_\alpha^I-E_1$.
In any case, we have the following variational characterization:
\begin{equation}\label{lambda}
  \lambda(\alpha,I)
  = \inf_{\psi \in \Dom(Q_{\alpha}^I) \setminus\{0\}}
  \frac{\, Q_\alpha^I[\psi]
  - E_1 \;\! \|\psi\|_{\sii(I\times\omega)}^2}
  {\|\psi\|_{\sii(I\times\omega)}^2}
  \,.
\end{equation}

It follows immediately from the Poincar\'e-type inequality
in the cross-section
\begin{equation}\label{Poincare}
  \|\nabla f\|_{\sii(\omega)}^2
  \, \geq \,
  E_1 \;\! \|f\|_{\sii(\omega)}^2
  \,, \qquad
  \forall f\in\sobi(\omega)
  \,,
\end{equation}
and Fubini's theorem
that $\lambda(\alpha,I)$ is non-negative.
In this subsection we establish a stronger, positivity result
provided that the twisting is effective in the following sense:
\begin{Lemma}\label{Lem.cornerstone}
Let $I \subset \Real$ be a bounded open interval.
Let~$\omega$ be not rotationally invariant with respect to the origin.
Let $\alpha \in L^\infty(I)$ be a non-trivial
(\ie, $\alpha\not=0$ on a subset of~$I$ of positive measure)
real-valued function.
Then
$$
  \lambda(\alpha,I) \geq \lambda_0
  \,,
$$
where $\lambda_0$ is a positive constant
depending on~$\|\alpha\|_{\sii(I)}$ and~$\omega$.
\end{Lemma}
\begin{proof}
We proceed by contradiction and assume that $\lambda(\alpha,I)=0$.
Since the spectrum of $H_\alpha^I$ is purely discrete,
the infimum in~\eqref{lambda} is attained
by a (smooth) function $\psi \in \Dom(Q_\alpha^I)$ satisfying
(recall~\eqref{Poincare})
\begin{equation}\label{2eqs}
  \|\partial_1\psi-\alpha\,\partial_u\psi\|_{\sii(I\times\omega)}^2
  = 0
  \qquad\mbox{and}\qquad
  \|\nabla'\psi\|_{\sii(I\times\omega)}^2
  - E_1 \;\! \|\psi\|_{\sii(I\times\omega)}^2
  = 0
  \,.
\end{equation}
Writing $\psi(s,t) = \varphi(s)\mathcal{J}_1(t) + \phi(s,t)$,
where~$\mathcal{J}_1$ is the positive eigenfunction of $-\Delta_D^\omega$
corresponding to~$E_1$ and
$
  (\mathcal{J}_1,\phi(s,\cdot))_{\sii(\omega)} = 0
$
for every $s \in I$,
we deduce from the second equality in~\eqref{2eqs} that $\phi=0$.
The first identity in~\eqref{2eqs} is then equivalent to
$$
  \|\dot\varphi\|_{\sii(I)}^2 \|\mathcal{J}_1\|_{\sii(\omega)}^2
  + \|\alpha \varphi\|_{\sii(I)}^2 \|\partial_u\mathcal{J}_1\|_{\sii(\omega)}^2
  - 2 (\mathcal{J}_1,\partial_u\mathcal{J}_1)_{\sii(\omega)}
  \Re (\dot\varphi,\alpha\varphi)_{\sii(I)}
  = 0
  \,.
$$
Since $(\mathcal{J}_1,\partial_u\mathcal{J}_1)_{\sii(\omega)}=0$
by an integration by parts, it follows that~$\varphi$ must be constant
and that
$$
  \|\alpha\|_{\sii(I)} = 0
  \qquad\mbox{or}\qquad
  \|\partial_u\mathcal{J}_1\|_{\sii(\omega)} = 0
  \,.
$$
However, this is impossible under the stated assumptions because
$\|\alpha\|_{\sii(I)}$ vanishes if and only if
$\alpha=0$ almost everywhere in~$I$,
and $\partial_u\mathcal{J}_1 = 0$ identically in~$\omega$
if and only if $\omega$~is rotationally invariant with respect to the origin.
\end{proof}
\begin{Remark}[An upper bound to the spectral threshold]
Irrespectively of whether the tube is twisted or not,
we always have the upper bound
\begin{equation}\label{upper}
  \lambda(\alpha,I) \leq \lambda\big(\|\alpha\|_{L^\infty(I)}\big)
  \,,
\end{equation}
where $\lambda(\alpha_0)$ denotes the first eigenvalue
of the operator~$B_{\alpha_0}$ in $\sii(\omega)$
associated with the quadratic form
$$
  b_{\alpha_0}[f] :=
  \|\nabla f\|_{\sii(\omega)}^2
  - E_1 \;\! \|f\|_{\sii(\omega)}^2
  + \alpha_0^2 \, \|
  \partial_u f
  \|_{\sii(\omega)}^2
  \,, \quad
  f\in\Dom(b_{\alpha_0}) := \sobi(\omega)
  \,.
$$
(Here $\partial_u$ is understood as a differential expression in~$\omega$.)
This can be seen easily by using the eigenfunction~$f_{\alpha_0}$ of~$B_{\alpha_0}$
corresponding to $\lambda(\alpha_0)$ with $\alpha_0 := \|\alpha\|_{L^\infty(I)}$
as a test function for $H_\alpha^I - E_1$.
More precisely, for any~$I$,
set $\psi(s,t):=\varphi_n(s)f_{\alpha_0}(t)$,
where~$\varphi_n$ is the restriction to~$I$ of the function
from the proof of Theorem~\ref{Thm.bending}.
Putting~$\psi$ into the Rayleigh quotient of~\eqref{lambda},
estimating
\begin{equation*}
  \|\partial_1\psi-\alpha\,\partial_u\psi\|_{\sii(I\times\omega)}^2
  \leq
  \left(
  \|\partial_1\psi\|_{\sii(I\times\omega)}
  + \alpha_0 \|\partial_u\psi\|_{\sii(I\times\omega)}
  \right)^2
\end{equation*}
and sending~$n$ to infinity, we conclude with~\eqref{upper}.

Note that $\lambda(\alpha_0)$ is non-negative due to~\eqref{Poincare}.
In fact, $\lambda(\alpha_0)=0$ if and only if $\alpha_0 = 0$
or~$\omega$ is rotationally invariant with respect to the origin.
(To show the positivity one can proceed
as in the proof of Lemma~\ref{Lem.cornerstone},
while the converse implication readily follows
by using~$\mathcal{J}_1$ as test function for~$B_{\alpha_0}$).
Consequently, we see that the hypotheses of Lemma~\ref{Lem.cornerstone}
represent also a necessary condition for the positivity of
$\lambda(\alpha,I)$.
\end{Remark}
\begin{Remark}
Note carefully that only bounded tubes
are allowed in Lemma~\ref{Lem.cornerstone}.
For instance, if $I=\Real$ then for any $\alpha$~vanishing at infinity,
$\lambda(\alpha,I)=0$ (\cf~Theorem~\ref{Thm.ess}).
\end{Remark}
\begin{Remark}[Periodically twisted tubes]\label{Rem.global}
As an example of unbounded tubes for which Lemma~\ref{Lem.cornerstone} still holds,
let us consider the case $I=\Real$ and $\alpha(s)=\alpha_0$ for a.e.\ $s\in\Real$.
Then there is an equality in~\eqref{upper},
\ie, $\lambda(\alpha_0,\Real) = \lambda(\alpha_0)$.
We prove it as follows.
Let~$\psi$ be any test function from $C_0^\infty(\Real\times\omega)$,
a dense subspace of $\Dom(Q_{\alpha_0}^\Real)$.
We employ the decomposition
$$
  \psi(s,t) = \phi(s,t) \, f_{\alpha_0}(t) \,,
  \qquad
  (s,t) \in \Real\times\omega \,,
$$
where~$f_{\alpha_0}$ is the eigenfunction of~$B_{\alpha_0}$
corresponding to $\lambda(\alpha_0)$
(we shall denote by the same letter
the function $1\otimes f_{\alpha_0}$ on $\Real \times \omega$)
and~$\phi$ is a function from $C_0^\infty(\Real\times\omega)$
actually introduced by this decomposition.
Then
\begin{multline*}
  Q_{\alpha_0}^\Real[\psi] =
  \lambda(\alpha_0) \;\! \|\psi\|_{\sii(\Real\times\omega)}^2
  - 2 \alpha_0 \, \Re\big(
  (\partial_1\phi)f_{\alpha_0},\phi\,\partial_u f_{\alpha_0}
  \big)_{\sii(\Real\times\omega)}
  \\
  + \|(\nabla'\phi)f_{\alpha_0}\|_{\sii(\Real\times\omega)}^2
  + \|
  (\partial_1\phi - \alpha_0\,\partial_u\phi)f_{\alpha_0}
  \|_{\sii(\Real\times\omega)}^2
  \,.
\end{multline*}
Neglecting the positive terms in the second line
and noticing that the mixed term is actually
equal to zero by an integration by parts,
we thus get $\lambda(\alpha_0,I) \geq \lambda(\alpha_0)$.
This together with~\eqref{upper} proves the desired equality.
An alternative proof,
based on a Floquet-type decomposition of~$H_{\alpha_0}^\Real$,
can be found in~\cite{EKov_2005}.

Note that the present proof can be readily adapted to show that
$\lambda^D(\alpha_0,I) \geq \lambda(\alpha_0)$
for any interval~$I$, where $\lambda^D(\alpha_0,I)$ is
the first eigenvalue of the operator in $\sii(I\times\omega)$
associated with the quadratic form which acts
in the same way as $Q_{\alpha_0}^I-E_1$
but has a smaller domain $\sobi(I\times\omega)$.
\end{Remark}
%

\subsection{Local and global Hardy inequalities}
%
Lemma~\ref{Lem.cornerstone} is the cornerstone of our method
to establish the existence of Hardy inequalities in twisted tubes.
For instance, the following Theorem gives a non-trivial inequality
provided that the hypotheses of Lemma~\ref{Lem.cornerstone} hold
on a subinterval of~$I$.
\begin{Theorem}\label{Thm.Hardy}
Let $\alpha \in L^\infty(I)$ be real-valued
and $I\subseteq\Real$ an open interval.
Let $\{I_j\}_{j \in K}$ be any collection of disjoint open
subintervals of~$I$, $K \subseteq \Nat$.
Then
\begin{equation}\label{Hardy}
  H_\alpha^I - E_1
  \ \geq \
  \sum_{j \in K} \,\lambda(\alpha,I_j) \, 1_{I_j}
\end{equation}
in the sense of quadratic forms.
Here $1_{I_j}$ denotes the operator of multiplication
by the characteristic function of $I_j\times\omega$.
\end{Theorem}
\begin{proof}
For every $\psi \in \Dom(Q_{\alpha}^I)$, we have

\begin{equation*}
  Q_\alpha^{I}[\psi] - E_1 \;\! \|\psi\|_{\sii(I\times\omega)}^2
  \geq \!
  \sum_{j \in K} \left[
  Q_\alpha^{I_j}[\psi] - E_1 \;\! \|\psi\|_{\sii(I_j\times\omega)}^2
  \right]
  \geq \!
  \sum_{j \in K} \lambda(\alpha,I_j)
  \;\! \|\psi\|_{\sii(I_j\times\omega)}^2
  .
\end{equation*}
Here the first inequality follows by~\eqref{Poincare}
with help of Fubini's theorem
and by the inclusion $I \supseteq \cup_{j \in K} I_j$.
The second inequality uses the variational definition~\eqref{lambda}
together with the trivial fact that
the restriction to $I_j\times\omega$
of a function from $\Dom(Q_{\alpha}^I)$
belongs to $\Dom(Q_{\alpha}^{I_j})$.
\end{proof}

Theorem~\ref{Thm.Hardy} has important consequences.
\begin{Corollary}\label{Cor.uniform}
Let $\alpha \in L^\infty(I)$ be real-valued and
$I \subseteq \Real$ an (arbitrary) open interval (bounded or unbounded).
Suppose that~$\omega$ is not rotationally invariant with respect to the origin
and that there exists a positive number~$\alpha_0$ such that
$|\alpha(s)| \geq \alpha_0$ for a.e.\ $s \in I$.
Then
$$
  \inf\sigma(H_\alpha^I) > E_1
  \,.
$$
\end{Corollary}
\begin{proof}
For bounded tubes this is already stated in Lemma~\ref{Lem.cornerstone}.
In general, it follows from Theorem~\ref{Thm.Hardy} by covering~$I$
by a sequence of disjoint bounded open subintervals~$I_j$
having the same length~$|I_j|$, \ie,
$
  I = \mathrm{int}
  \left(
  \cup_{j \in K} \overline{I_j}
  \right)
$.
Indeed, Theorem~\ref{Thm.Hardy} yields
$
  H_\alpha^I - E_1 \geq \inf_{j \in K} \lambda(\alpha,I_j)
$.
By Lemma~\ref{Lem.cornerstone},
each $\lambda(\alpha,I_j)$ can be estimated from below
by a positive number~$\lambda_0^j$
which depends uniquely on $\|\alpha\|_{\sii(I_j)}$ and~$\omega$.
However, using the trivial bounds
$
  \alpha_0^2 \, |I_j|
  \leq \|\alpha\|_{\sii(I_j)}^2
  \leq \|\alpha\|_{L^\infty(I)}^2 \, |I_j|
$
and recalling the uniform length of~$I_j$'s,
we see that each~$\lambda_0^j$ can be estimated from below
by a positive number independent of~$j$.
\end{proof}

Recalling that~$H_{\dot\theta}^\Real$ coincides with~$H$,
which is unitarily equivalent to the Dirichlet Laplacian
in a twisted tube, we see that the never-vanishing twisting
rises the spectral threshold
(periodically twisted tubes of Remark~\ref{Rem.global}
are just one example).

However, the case of particular interest corresponds to $I=\Real$
with~$\alpha$ vanishing at infinity.
Then~$E_1$ corresponds to the threshold of the (essential) spectrum
of~$H_\alpha^\Real$ (\cf~Theorem~\ref{Thm.ess})
and~\eqref{Hardy} may be referred to as
a Hardy inequality for $H_\alpha^\Real-E_1$.
If~$\alpha$ is of definite sign on~$\Real$
and~$\omega$ is not rotationally invariant with respect to the origin,
then~\eqref{Hardy} is \emph{global} in the sense
that its right hand side
(with bounded $I_j$'s covering~$\Real$
as in the proof of Corollary~\ref{Cor.uniform})
represents a positive Hardy weight vanishing at infinity only;
the rate in which it goes to zero at infinity
is determined by asymptotic properties of~$\alpha$.
On the other hand, if~$\alpha$ equals zero
outside a bounded interval, then~\eqref{Hardy} is \emph{local}
since the Hardy weight vanishes outside the interval too.

In the latter case, however, Theorem~\ref{Thm.Hardy}
implies the following global Hardy inequality:
\begin{Theorem}\label{Thm.Hardy.global}
Let $I\subseteq\Real$ be an open interval.
Let~$\omega$ be not rotationally invariant with respect to the origin.
Let $\alpha \in L^\infty(I)$ be a non-trivial real-valued function
of compact support in~$I$.
Then
$$
  H_\alpha^I - E_1
  \ \geq \
  \frac{c}{1+\delta^2}
$$
in the sense of quadratic forms.
Here~$\delta(s,t):=|s-s_0|$,
$(s,t) \in I\times\omega$,
$s_0$ is the mid-point of the interval
$J := (\inf\supp \alpha,\sup\supp\alpha)$,
and~$c$ is a positive constant depending on~$\alpha$ and~$\omega$.
\end{Theorem}
\begin{proof}
For clarity of the exposition,
we divide the proof into several steps.

\smallskip
\noindent
1.~The main ingredient in the proof is the following
Hardy-type inequality for a Schr\"odinger operator
in $I\times\omega$
with a characteristic-function potential:
\begin{equation}\label{Hardy.classical}
  \|(1+\delta^2)^{-1/2}\psi\|_{\sii(I\times\omega)}^2
  \leq 16 \, \|\partial_1\psi\|_{\sii(I\times\omega)}^2
  + (2+64/|J|^2) \, \|\psi\|_{\sii(J\times\omega)}^2
\end{equation}
for every $\psi \in \Sobi(I\times\omega)$ and $ J\subset I$.
This inequality can be established quite easily (\cf~\cite[Sec.~3.3]{EKK})
by means of Fubini's theorem
and the classical one-dimensional Hardy inequality
$
  \int_0^b s^{-2} |\varphi(s)|^2 ds
  \leq 4 \int_0^b |\dot\varphi(s)|^2 ds
$
valid for any $\varphi\in\Sobi((0,b))$, $b>0$,
satisfying $\varphi(0)=0$.

\smallskip
\noindent
2.~By Theorem~\ref{Thm.Hardy},
we have
\begin{equation}\label{bound1}
  Q_{\alpha}^{I}[\psi] - E_1 \;\! \|\psi\|_{\sii(I\times\omega)}^2
  \, \geq \,
  \lambda(\alpha,J) \, \|\psi\|_{\sii(J\times\omega)}^2
\end{equation}
for every $\psi \in \Dom(Q_{\alpha}^{I})$.
Under the stated hypotheses,
we know from Lemma~\ref{Lem.cornerstone} that
$\lambda(\alpha,J)$ is a positive number.

\smallskip
\noindent
3.~Finally, for sufficiently small $\eps\in(0,1)$
and every $\psi \in \Dom(Q_{\alpha}^{I})$,
we have
\begin{eqnarray}\label{bound3}
\lefteqn{
   Q_{\alpha}^{I}[\psi] - E_1 \;\! \|\psi\|_{\sii(I\times\omega)}^2
   }
   \nonumber \\
   && \geq
   \eps \, \|\partial_1\psi\|_{\sii(I\times\omega)}^2
   + \|\nabla'\psi\|_{\sii(I\times\omega)}^2
   - E_1 \;\! \|\psi\|_{\sii(I\times\omega)}^2
   - \frac{\eps}{1-\eps} \, \|\alpha\partial_u\psi\|_{\sii(I\times\omega)}^2
   \nonumber \\
   && \geq
   \eps \, \|\partial_1\psi\|_{\sii(I\times\omega)}^2
   - \frac{\eps}{1-\eps} \, \|\alpha\|_{L^\infty(I)}^2 \, a^2 \, E_1
   \;\! \|\psi\|_{\sii(J\times\omega)}^2
   \,.
\end{eqnarray}
The first inequality is due to elementary estimates
of the mixed term of $Q_{\alpha}^{I}[\psi]$
based on the Schwarz and Cauchy inequalities.
The second inequality follows by a consecutive use
of the pointwise bound
$
  |\alpha(s) \partial_u\psi(s,t)|
  \leq \|\alpha\|_{L^\infty(I)} \, \chi_{J}(s) \, a \, |\nabla'\psi(s,t)|
$
and the inequality~\eqref{Poincare},
with help of Fubini's theorem;
it holds provided~$\eps$ is less than $(1+a^2 \|\alpha\|_{L^\infty(I)}^2)^{-1}$.

\smallskip
\noindent
4.~Interpolating between the bounds~\eqref{bound1} and~\eqref{bound3},
and using~\eqref{Hardy.classical} in the latter,
we finally arrive at
\begin{multline*}
   Q_{\alpha}^{I}[\psi] - E_1 \;\! \|\psi\|_{\sii(I\times\omega)}^2
   \geq
   \frac{1}{2}
   \frac{\eps}{16} \
   \|(1+\delta^2)^{-1/2}\psi\|_{\sii(I\times\omega)}^2
   \\
   +
   \frac{1}{2}
   \left[
   \lambda(\alpha,J)
   - \eps \, \left(\frac{1}{8}+\frac{4}{|J|^2}\right)
   - \frac{\eps}{1-\eps} \, \|\alpha\|_{L^\infty(I)}^2 \, a^2 E_1
   \right]
   \|\psi\|_{\sii(J\times\omega)}^2
\end{multline*}
for every $\psi \in \Dom(Q_{\alpha}^{I})$.
It is clear that the last line on the right hand side of this inequality
can be made non-negative by choosing~$\eps$ sufficiently small.
Such an~$\eps$ then determines the Hardy constant~$c$.
\end{proof}
%

\subsection{Historical remarks}
%
The tool of Hardy inequalities in a quantum-wave\-guide context
was used for the first time by Ekholm and Kova\v{r}{\'\i}k in \cite{MK-Kov}
(see also \cite{B-MK-Kov}).
In their papers, the existence of a magnetic Hardy-type inequality
in a two-dimensional strip was established and used to prove
certain stability of the spectrum against geometric
or boundary-condition perturbations.

The existence of a Hardy inequality in twisted tubes
was first conjectured by Timo Weidl in a private communication
after my talk in Giens~\cite{K-Giens} based on~\cite{EFK},
and proved subsequently by Ekholm, Kova\v{r}{\'\i}k
and the present author in~\cite{EKK}.
In this reference,
variants of Theorems~\ref{Thm.Hardy} and~\ref{Thm.Hardy.global}
were established under additional assumptions about~$\alpha$
(it had to be a continuous function with bounded derivatives).
In the present paper, we give a new proof,
which we believe is more elegant and straightforward,
and which enables us to relax the technical hypotheses about~$\alpha$.

Exner and Kova\v{r}{\'\i}k studied in~\cite{EKov_2005}
the case of a periodically twisted tube
(\ie\ constant~$\alpha$)
and demonstrated, \emph{inter alia}, that the spectral
threshold of~$H_{\alpha}^\Real$ starts strictly above~$E_1$
(\cf~Remark~\ref{Rem.global}).
As the main result, it was shown that a local perturbation
of the periodically twisted tube
may give rise to eigenvalues below the essential spectrum
(\cf~also~\cite{Exner-Fraas_2007}).

The initial paper~\cite{EKK} has been followed by a couple
of subsequent works in which the robustness of the spectral-geometric
effect due to twisting has been clearly demonstrated in other models, too.
In particular:

\smallskip
\noindent
1.~\underline{Twisted strips}.
The present author demonstrated in~\cite{K3}
that similar Hardy inequalities exist also
for the Dirichlet Laplacian in two-dimensional strips
twist\-ed in the three-dimensional Euclidean space.
Such strips can be reconsidered as being embedded in ruled surfaces~\cite{K1}
(\cf~also Figure~\ref{f-strip}).
Then the Hardy inequalities
are effectively induced by negative Gauss curvature.
Roughly speaking, negative curvature of the ambient space
acts as a repulsive interaction.

\smallskip
\noindent
2.~\underline{Twisted boundary conditions}.
Of course, there is no geometrical twist for planar strips.
However, one can still introduce some sort of twisting
by changing boundary conditions appropriately.
This was done by Kova\v{r}{\'\i}k and the present author in~\cite{KK2}.
In that paper, a Hardy inequality was demonstrated for the Laplacian
in a straight strip, subject to a combination
of Dirichlet and Neumann boundary conditions (see Figure~\ref{f-DN}),
and used to prove an absence of discrete eigenvalues
in the case when the Neumann boundary conditions overlap.
(In fact, the model itself was initially introduced by
Dittrich and K\v{r}{\'\i}\v{z} in~\cite{DKriz1},
where the latter result was proved by a different method.)

\begin{figure}[!h]
\begin{center}
\includegraphics[width=0.95\textwidth]{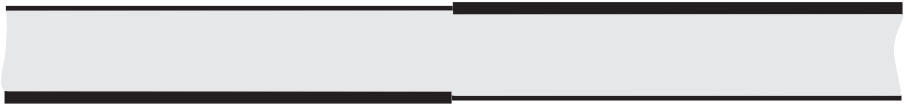}
\caption{Twisting in the two-dimensional model of~\cite{KK2}
introduced via switching Dirichlet (thick lines)
to Neumann (thin lines) boundary conditions at one point,
and \emph{vice versa}.}\label{f-DN}
\end{center}
\end{figure}
%


Let us also mention the recent work~\cite{Kovarik-Sacchetti_2007}
where a repulsive effect of twisting
on embedded eigenvalues is demonstrated.

As for non-twisted models,
in addition to the magnetic strips mentioned above,
the tool of Hardy inequalities was used also in curved waveguides
with a combination of (non-twisted) Dirichlet and Robin boundary conditions
\cite{FK3}.

In contrast to the waveguide (or more generally quasi-cylindrical) case,
there is an extensive literature on the Hardy inequalities
in bounded domains or complete manifolds;
we refer to the review article~\cite{Davies_1999}.

\section{Twisting contra mild bending}\label{Sec.mild}
%
It follows from Theorem~\ref{Thm.bending} that any bending,
no matter how small, generates a spectrum of $-\Delta_D^\Omega$
below the energy~$E_1$ provided that the tube is not twisted.
As an application of the Hardy inequality of Theorem~\ref{Thm.Hardy.global},
we prove now that whenever the tube is twisted
some critical strength of the bending is needed
in order to induce a spectrum below~$E_1$.
\begin{Theorem}\label{Thm.mild}
Let~$\omega$ be not rotationally invariant with respect to the origin.
Let $\tau-\dot\theta$ be a non-trivial function of compact support.
Assume that for all $s\in\Real$,
$$
  |\kappa(s)| \leq \eps(s) := \frac{\eps_0}{1+s^2}
  \qquad\mbox{with}\quad
  \eps_0 \geq 0
  \,.
$$
Then there exists a positive number~$\eps_0^*$ such that
$\eps_0 \leq \eps_0^*$ implies that
$$
  \inf\sigma(-\Delta_D^\Omega) \geq E_1
  \,.
$$
Here $\eps_0^*$ depends on $\tau-\dot\theta$ and~$\omega$.
\end{Theorem}
\begin{proof}
Let~$\psi$ belong to $C_0^\infty(\Real\times\omega)$,
a dense subspace of $\Dom(Q)$.
The proof is based on an algebraic comparison
of $Q[\psi]-E_1\|\psi\|_{\Hilbert}^2$
with $Q_{\dot{\theta}-\tau}^\Real[\psi]-E_1\|\psi\|_{\sii(\Real\times\omega)}^2$
and the usage of Theorem~\ref{Thm.Hardy.global}.
Let~$G_0$ be the matrix~\eqref{metric} after letting $\kappa=0$;
then
$
  Q_{\dot{\theta}-\tau}^\Real[\psi]
  = (\partial_i\psi,G_0^{ij}\partial_j\psi)_{\sii(\Real\times\omega)}
$.
First we note that
$$
  \forall (s,t)\in\Real\times\omega \,, \qquad
  h_-(s) :=
  1 - a |\kappa(s)|
  \leq h(s,t) \leq
  1 + a |\kappa(s)|
  =: h_+(s)
  \,.
$$
Second, since $G^{-1}G_0$ is a matrix
with eigenvalues~$1$ (double) and $h^{-2}$,
we also have
$$
  \forall (s,t)\in\Real\times\omega \,, \qquad
  G^{-1}(s,t) \geq h_+^{-2}(s) \ G_0^{-1}(s,t)
$$
in the sense of matrices.
Consequently,
\begin{eqnarray}\label{Hardy.mild}
\lefteqn{
  Q[\psi]-E_1\|\psi\|_{\Hilbert}^2
  }
  \nonumber \\
  && \geq \int_{\Real}
  \frac{h_-(s)}{h_+^2(s)}
  \left[
  \big(
  \partial_i\psi(s,\cdot),G_0^{ij}(s,\cdot)\partial_j\psi(s,\cdot)
  \big)_{\sii(\omega)}
  - E_1 \;\! \|\psi(s,\cdot)\|_{\sii(\omega)}^2
  \right]
  ds
  \nonumber \\
  && \phantom{\geq} + E_1 \int_\Real
  g(s) \,
  \|\psi(s,\cdot)\|_{\sii(\omega)}^2 \,
  ds
  \qquad\mbox{with}\qquad
  g := \frac{h_-}{h_+^2} - h_+
  \nonumber \\
  && \geq
  \frac{1-a\|\kappa\|_{L^\infty(\Real)}}
  {(1+a\|\kappa\|_{L^\infty(\Real)})^2}
  \left(
  Q_{\dot{\theta}-\tau}^\Real[\psi]-E_1\|\psi\|_{\sii(\Real\times\omega)}^2
  \right)
  \nonumber \\
  && \phantom{\geq}
  + E_1 \int_\Real
  g(s) \,
  \|\psi(s,\cdot)\|_{\sii(\omega)}^2 \,
  ds
  \nonumber \\
  && \geq
  \int_\Real \left[
  \frac{1-a\|\kappa\|_{L^\infty(\Real)}}
  {(1+a\|\kappa\|_{L^\infty(\Real)})^2}
  \frac{c}{1+(s-s_0)^2} + E_1 \;\! g(s)
  \right]
  \|\psi(s,\cdot)\|_{\sii(\omega)}^2
  ds
  \,,
\end{eqnarray}
where the last expression is non-negative
for sufficiently small~$\eps_0$
because $g(s)=\mathcal{O}(s^{-2})$ as $|s| \to \infty$
due to the stated assumption about~$\kappa$.
In the second inequality we have used~\eqref{Poincare}.
The last inequality follows by Theorem~\ref{Thm.Hardy.global}
with~$s_0$ being the mid-point of the interval
$\big(\inf\supp(\tau-\dot\theta),\sup\supp(\tau-\dot\theta)\big)$.
\end{proof}

As a direct consequence of Theorems~\ref{Thm.mild} and~\ref{Thm.ess},
we get that the spectrum $[E_1,\infty)$ is stable as a set:
\begin{Corollary}
Under the hypotheses of Theorem~\ref{Thm.mild},
we have for all $\eps_0 \leq \eps_0^*$,
$$
  \sigma(-\Delta_D^\Omega) = [E_1,\infty)
  \,.
$$
\end{Corollary}

A variant of Theorem~\ref{Thm.mild} was proved
by Ekholm, Kova\v{r}{\'\i}k and the present author in~\cite{EKK}
under the additional assumptions that~$\kappa$
was compactly supported and that the supremum norm
of the derivative~$\dot\kappa$ was sufficiently small too.
The simpler proof we perform in the present paper
is inspired by an idea of~\cite{K3}.
Notice that the statement of Theorem~\ref{Thm.mild}
follows as a consequence of a stronger result,
namely the Hardy-type inequality~\eqref{Hardy.mild}.

If~$\Omega$ is mildly bent but not twisted,
$-\Delta_D^\Omega$ possesses weakly-coupled eigenvalues
whose asymptotic properties were initially studied
by Duclos and Exner in~\cite{DE}.
Further results, namely sufficient and necessary conditions
for the existence of the weakly-coupled eigenvalues
in simultaneously mildly twisted and bent tubes,
were obtained only recently by Grushin in~\cite{Grushin_2005}
(\cf~also~\cite{Grushin_2004}).
The influence of twisting on \emph{embedded} eigenvalues
of Schr\"odinger operators in mildly twisted tubes without bending
were analysed in another recent work~\cite{Kovarik-Sacchetti_2007}
by Kova\v{r}{\'\i}k and Sacchetti.

\section{Conclusions}\label{Sec.end}
%
Motivated by mesoscopic physics,
in this paper we were interested in the interplay between
the geometry of a three-dimensional tube
and spectral properties of the associated Dirichlet Laplacian.
The moral of our study is as follows:
\begin{enumerate}
\item
\emph{bending acts as an attractive interaction};
\item
\emph{twisting acts as a repulsive interaction}.
\end{enumerate}
The effect of bending has been known for almost two decades
\cite{ES,GJ,DE,ChDFK}.
On the other hand, the effect of twisting is a very recent result,
based on an existence of Hardy-type inequalities in twisted tubes~\cite{EKK}.
The main goal of the present work was to revise and improve
the original results of~\cite{EKK},
and provide a self-contained publication
on the two independent spectral-geometric effects
of bending and twisting.

Let us conclude the paper with some open problems:

\smallskip
\noindent
1.~\underline{Higher-dimensional generalizations}.
It is well known that bending of a tube about
a complete non-compact surface
in the three-dimensional Euclidean space
acts as an attractive interaction
\cite{DEK2,CEK,LL1,LL2,LL3}.
On the other hand, by analogy with the present model,
we expect the existence of Hardy-type inequalities
if the codimension of the reference manifold increases.
Prove it.

\smallskip
\noindent
2.~\underline{Effect of twisting on the essential spectrum}.
A detailed study of the nature of the essential spectrum
in non-twisted bent tubes was performed in~\cite{KT} via the Mourre theory.
Can one improve the analysis by using
the existence of Hardy inequalities in twisted tubes?
A repulsive effect of twisting on eigenvalues embedded
in the essential spectrum in non-bent tubes
was demonstrated recently in~\cite{Kovarik-Sacchetti_2007}.

\smallskip
\noindent
3.~\underline{An optimization problem}.
For simplicity, let us assume that~$\omega$
is the disc~$B_a$ of fixed radius~$a$ centred at the origin of~$\Real^2$.
In~\cite{EFK} Exner, Freitas and the present author proved
that the inequality
\begin{equation}\label{EFK-inequality}
  \lambda_1(\Omega) := \inf\sigma(-\Delta_D^\Omega)
  \ \geq \
  \lambda_1(\Omega^*)
\end{equation}
holds for every tube~$\Omega$ (bounded or unbounded)
about any curve~$\Gamma$ satisfying~\eqref{Ass.basic1}.
Here $\lambda_1(\Omega^*)$ denotes the first eigenvalue
of the Dirichlet Laplacian in the toroidal tube~$\Omega^*$ obtained
by revolving~$B_a$ about an axis at the distance $(\sup\kappa)^{-1}$
from the centre of~$B_a$ (with the convention that~$\Omega^*$ is
a straight tube if $\kappa=0$).
Clearly, the inequality~\eqref{EFK-inequality} is optimal
in the sense that the equality is achieved for a tube geometry.
However, the question about an optimal lower bound
under the additional constraint that~$\Omega$ is \emph{unbounded}
is more difficult and remains open.

\smallskip
\noindent
4.~\underline{Thin tubes}.
The following beautiful result provides an insight
into the mechanism behind our qualitative results,
at least in the regime of thin tubes:
\begin{Theorem}[Bouchitt\'e, Mascarenhas and Trabucho \cite{BMT}]\label{Thm.thin}
Let $\Omega_\eps:=\mathcal{L}(I\times\eps\omega)$,
where $I\subset\Real$ is a bounded open interval
and $\eps\omega:=\{\eps t \,|\, t\in\omega\}$, $\eps>0$.
Let $\{\lambda_j(\eps)\}_{j=1}^\infty$ be the non-decreasing
sequence of eigenvalues (repeated according to multiplicities)
of the Dirichlet Laplacian $-\Delta_D^{\Omega_\eps}$ in $\sii(\Omega_\eps)$.
Then
$$
  \lambda_j(\eps) = \frac{E_1}{\eps^2} + \mu_j + o(1)
  \qquad\mbox{as}\qquad \eps \to 0
  \,,
$$
where $\{\mu_j\}_{j=1}^\infty$ denotes
the non-decreasing sequence of eigenvalues
(repeated according to multiplicities)
of the Schr\"odinger operator
$$
  -\Delta_D^{I} - \frac{\kappa^2}{4} + C(\omega) (\tau-\dot\theta)^2
  \qquad\mbox{in}\qquad
  \sii(I)
  \,.
$$
Here $C(\omega)$ is a non-negative constant depending uniquely on~$\omega$.
Moreover, $C(\omega)>0$ if, and only if,
$\omega$ is not rotationally invariant with respect to the origin.
\end{Theorem}
\noindent
Although Theorem~\ref{Thm.thin} was proved in~\cite{BMT}
for bounded tubes only, we believe that the convergence
results extend to eigenvalues below the essential spectrum
in unbounded tubes as well. Prove it.
In the case of non-twisted tubes
(\ie, $\tau-\dot\theta=0$ or $C(\omega)=0$),
Theorem~\ref{Thm.thin} has been known for several years
\cite{DE,FK4}, including unbounded tubes.
It would be also desirable to extend the study
of the nodal set of eigenfunctions performed in~\cite{FK4}
to twisted tubes.

\smallskip
\noindent
5.~\underline{Other physical models}.
Is the repulsive/stabilizing effect of twisting limited
to the specific model of the present paper
(related to the electronic transport
in quantum heterostructures)?
Or, does it have a counterpart in other physical situations,
namely in electromagnetism or fluid mechanics?
More specifically, is there a Hardy-type inequality related
to the linear operator pencil studied in~\cite{JLP}?

%

\begin{thebibliography}{CDFK05}

\bibitem[BEK05]{B-MK-Kov}
D.~Borisov, T.~Ekholm, and H.~Kova{\v{r}}{\'\i}k, \emph{Spectrum of the
  magnetic {S}chr{\"o}\-ding\-er operator in a waveguide with combined boundary
  conditions}, Ann. H.~Poincar{\'e} \textbf{6} (2005), 327--342.

\bibitem[BMT07]{BMT}
G.~Bouchitt\'e, M.~L. Mascarenhas, and L.~Trabucho, \emph{On the curvature and
  torsion effects in one dimensional waveguides}, Control, Optimisation and
  Calculus of Variations \textbf{13} (2007), no.~4, 793--808.

\bibitem[CB96a]{CB1}
  I.~J. Clark and A.~J. Bracken,
  \emph{Effective potentials of quantum strip waveguides and their
  dependence upon torsion}, J.~Phys.~A \textbf{29} (1996), 339--348.

\bibitem[CB96b]{CB2}
  \bysame, \emph{Bound states in tubular quantum waveguides
  with torsion}, J.~Phys.~A \textbf{29} (1996), 4527--4535.

\bibitem[CDFK05]{ChDFK}
B.~Chenaud, P.~Duclos, P.~Freitas, and D.~Krej\v{c}i\v{r}\'{\i}k,
  \emph{Geometrically induced discrete spectrum in curved tubes}, Differential
  Geom. Appl. \textbf{23} (2005), no.~2, 95--105.

\bibitem[CdV06]{CdeV_2006}
Y.~Colin~de Verdi{\`e}re, \emph{Une introduction {\`a} la {``Phase de
  Berry''}}, available on:
  \url{http://www-fourier.ujf-grenoble.fr/~ycolver/}.

\bibitem[CEK04]{CEK}
G.~Carron, P.~Exner, and D.~Krej\v{c}i\v{r}\'{\i}k, \emph{Topologically
  nontrivial quantum layers}, J.~Math.\ Phys. \textbf{45} (2004), 774--784.

\bibitem[Dav99]{Davies_1999}
E.~B. Davies, \emph{A review of {H}ardy inequalities}, Oper. Theory Adv. Appl.,
  vol. 110, Birkha{\"u}ser, Basel-Boston, 1999, (Conference in honour of
  V.~G.~Maz'ya, Rostock, 1998), pp.~55--67.

\bibitem[DDI98]{DDI}
Y.~Dermenjian, M.~Durand, and V.~Iftimie, \emph{Spectral analysis of an
  acoustic multistratified perturbed cylinder}, Commun. in Partial Differential
  Equations \textbf{23} (1998), no.~1{\&}2, 141--169.

\bibitem[DE95]{DE}
P.~Duclos and P.~Exner, \emph{{C}urvature-induced bound states in quantum
  waveguides in two and three dimensions}, Rev. Math. Phys. \textbf{7} (1995),
  73--102.

\bibitem[DEK01]{DEK2}
P.~Duclos, P.~Exner, and D.~Krej\v{c}i\v{r}\'{\i}k, \emph{Bound states in
  curved quantum layers}, Commun. Math. Phys. \textbf{223} (2001), 13--28.

\bibitem[DK02]{DKriz1}
J.~Dittrich and J.~K{\v{r}}{\'{\i}}{\v{z}}, \emph{Bound states in straight
  quantum waveguides with combined boundary condition}, J. Math. Phys.
  \textbf{43} (2002), 3892--3915.

\bibitem[EF07]{Exner-Fraas_2007}
P.~Exner and M.~Fraas, \emph{A remark on helical waveguides}, Phys. Lett.~A
  \textbf{369} (2007), 393--399.

\bibitem[EFK04]{EFK}
P.~Exner, P.~Freitas, and D.~Krej\v{c}i\v{r}\'{\i}k, \emph{A lower bound to the
  spectral threshold in curved tubes}, R.~Soc. Lond. Proc. Ser.~A Math. Phys.
  Eng. Sci. \textbf{460} (2004), no.~2052, 3457--3467.

\bibitem[EK05a]{MK-Kov}
T.~Ekholm and H.~Kova{\v{r}}{\'\i}k, \emph{Stability of the magnetic
  {S}chr{\"o}dinger operator in a waveguide}, Commun. in Partial Differential
  Equations \textbf{30} (2005), no.~4, 539--565.

\bibitem[EK05b]{EKov_2005}
P.~Exner and H.~Kova\v{r}{\'\i}k, \emph{Spectrum of the Schr\"odinger operator in
  a perturbed periodically twisted tube}, Lett. Math. Phys. \textbf{73} (2005),
  183--192.

\bibitem[EKK08]{EKK}
T.~Ekholm, H.~Kova{\v{r}}{\'\i}k, and D.~Krej\v{c}i\v{r}\'{\i}k, \emph{A
  {H}ardy inequality in twisted waveguides}, Arch. Ration. Mech. Anal.
  \textbf{188} (2008), no. 2, 245--264.

\bibitem[E{\v S}89]{ES}
P.~Exner and P.~{\v S}eba, \emph{Bound states in curved quantum waveguides},
  J.~Math.~Phys. \textbf{30} (1989), 2574--2580.

\bibitem[FK06]{FK3}
\bysame, \emph{Waveguides with combined {D}irichlet and {R}obin boundary
  conditions}, Math. Phys. Anal. Geom. \textbf{9} (2006), no.~4, 335--352.

\bibitem[FK08]{FK4}
P.~Freitas and D.~Krej\v{c}i\v{r}\'{\i}k,
  \emph{Location of the nodal set for thin curved tubes},
  Indiana Univ. Math. J. \textbf{57} (2008), no. 1, 343--376.

\bibitem[GJ92]{GJ}
J.~Goldstone and R.~L. Jaffe, \emph{Bound states in twisting tubes}, Phys.
  Rev.~B \textbf{45} (1992), 14100--14107.

\bibitem[Gla65]{Glazman}
I.~M. Glazman, \emph{Direct methods of qualitative spectral analysis of
  singular differential operators}, Israel Program for Scientific Translations,
  Jerusalem, 1965.

\bibitem[Gru04]{Grushin_2004}
V.~V. Grushin, \emph{{On the eigenvalues of finitely perturbed Laplace
  operators in infinite cylindrical domains}}, Math. Notes \textbf{75} (2004),
  no.~3, 331--340.

\bibitem[Gru05]{Grushin_2005}
\bysame, \emph{{Asymptotic behavior of the eigenvalues of the Schr\"odinger
  operator with transversal potential in a weakly curved infinite cylinder}},
  Math. Notes \textbf{77} (2005), no.~5, 606--613.

\bibitem[JLP06]{JLP}
E.~R. Johnson, M.~Levitin, and L.~Parnovski, \emph{Existence of eigenvalues of
  a linear operator pencil in a curved waveguide -- localized shelf waves on a
  curved coast}, Siam J. Math. Anal. \textbf{37} (2006), 1465--1481.

\bibitem[Kat66]{Kato}
T.~Kato, \emph{Perturbation theory for linear operators}, Springer-Verlag,
  Berlin, 1966.

\bibitem[KdA04]{KT}
D.~Krej\v{c}i\v{r}\'{\i}k and R.~Tiedra de~Aldecoa, \emph{The nature of the
  essential spectrum in curved quantum waveguides}, J.~Phys.~A \textbf{37}
  (2004), no.~20, 5449--5466.

\bibitem[KK05]{KKriz}
D.~Krej\v{c}i\v{r}\'{\i}k and J.~K\v{r}\'{\i}\v{z}, \emph{On the spectrum of
  curved quantum waveguides}, Publ.~RIMS, Kyoto University \textbf{41} (2005),
  no.~3, 757--791.

\bibitem[KK08]{KK2}
H.~Kova\v{r}{\'\i}k and D.~Krej\v{c}i\v{r}\'{\i}k,
  \emph{{A Hardy inequality in a twisted Dirichlet-Neumann waveguide}},
  Math. Nachr. \textbf{281} (2008), no. 8, 1159--1168.

\bibitem[Kli78]{Kli}
W.~Klingenberg, \emph{A course in differential geometry}, Springer-Verlag, New
  York, 1978.

\bibitem[Kre03]{K1}
D.~Krej\v{c}i\v{r}\'{\i}k, \emph{Quantum strips on surfaces}, J.~Geom. Phys.
  \textbf{45} (2003), no.~1--2, 203--217.

\bibitem[Kre04]{K-Giens}
\bysame, \emph{A lower bound to the spectral threshold in curved tubes}, talk
  at the 9th international conference {``Mathematical Results in Quantum
  Theory''}, Giens, France, 12--16 September 2004.

\bibitem[Kre06]{K3}
\bysame, \emph{Hardy inequalities in strips on ruled surfaces}, J. Inequal.
  Appl. \textbf{2006} (2006), Article ID 46409, 10 pages.

\bibitem[Kre07]{K-Cambridge}
\bysame, \emph{Twisting versus bending in quantum waveguides}, talk at the
  international workshop {``Graph Models of Mesoscopic Systems, Wave-Guides and
  Nano-Structures''}, Cambridge, UK, 10--13 April 2007; available on
  \url{http://www.newton.ac.uk/programmes/AGA/seminars/041316301.html}.

\bibitem[KS07]{Kovarik-Sacchetti_2007}
H.~Kova{\v{r}}{\'\i}k and A.~Sacchetti, \emph{{Resonances in twisted quantum
  waveguides}}, J. Phys. A \textbf{40} (2007), 8371--8384.

\bibitem[LCM99]{LCM}
J.~T. Londergan, J.~P. Carini, and D.~P. Murdock, \emph{Binding and scattering
  in two-dimensional systems}, LNP, vol. m60, Springer, Berlin, 1999.

\bibitem[LL06]{LL2}
Ch. Lin and Z.~Lu, \emph{On the discrete spectrum of generalized quantum
  tubes}, Commun. in Partial Differential Equations \textbf{31} (2006), 1529--1546.

\bibitem[LL07a]{LL1}
\bysame, \emph{Existence of bound states for layers built over hypersurfaces in
  {$\mathbb{R}^{n+1}$}}, J.~Funct. Anal. \textbf{244} (2007), 1--25.

\bibitem[LL07b]{LL3}
\bysame, \emph{Quantum layers over surfaces ruled outside a compact set},
  J.~Math. Phys. \textbf{48} (2007), art. no. 053522.

\bibitem[Spi79]{Spivak2}
M.~Spivak, \emph{A comprehensive introduction to differential geometry},
  vol.~II, Publish or Perish, Houston, Texas, 1979.

\bibitem[TG89]{Tsao-Gambling_1989}
C.~Y.~H. Tsao and W.~A. Gambling, \emph{Curvilinear optical fibre waveguide:
  characterization of bound modes and radiative field}, Proc. R.~Soc. Lond.~A
  \textbf{425} (1989), 1--16.

\end{thebibliography}
%

\providecommand{\bysame}{\leavevmode\hbox to3em{\hrulefill}\thinspace}
\providecommand{\MR}{\relax\ifhmode\unskip\space\fi MR }
\providecommand{\MRhref}[2]{%
  \href{http://www.ams.org/mathscinet-getitem?mr=#1}{#2}
}
\providecommand{\href}[2]{#2}

\end{document}